\RequirePackage{fix-cm}

\documentclass[smallextended]{svjour3}       % onecolumn (second 
\smartqed  % flush right qed marks, e.g. at end of proof
\usepackage{graphicx}

\usepackage{tikz}
\usetikzlibrary{matrix}%
\usepackage[inline]{enumitem}
\usepackage{amssymb}
\usepackage[nofillcomment,vlined,linesnumbered,titlenumbered,algo2e,ruled]{algorithm2e}

\providecommand{\DontPrintSemicolon}{\dontprintsemicolon}
\providecommand{\LinesNumbered}{\linesnumbered}

\begin{document}

\title{Efficient enumeration of chordless cycles\thanks{The first author was partially supported by FAPEG -- Funda\c{c}\~{a}o de Amparo \`{a} Pesquisa do Estado de Goi\'{a}s. The last author was supported by CAPES -- Coordena\c{c}\~{a}o de Aperfei\c{c}oamento de Pessoal de N\'{i}vel Superior.}
}

\author{
        Elis\^{a}ngela S. Dias%
   \and Diane Castonguay%
   \and Humberto Longo%
   \and Walid A. R. Jradi%
        }

\institute{Elis\^{a}ngela S. Dias \and Diane Castonguay \and Humberto Longo \and Walid A.R. Jradi \at
              Instituto de Inform\'{a}tica, Universidade Federal de Goi\'{a}s, Campus Samambaia, Goi\^{a}nia, Goi\'{a}s, Brazil \\
              Tel.: +55-62-35211181\\
              Fax: +55-62-35211182\\
              \email{\{elisangela, diane, longo, walid.jradi\}@inf.ufg.br}           %  \\
}

\date{Received: date / Accepted: date}

\maketitle

\begin{abstract}
%In a finite undirected simple graph, a {\it chordless cycle} is an induced subgraph which is a cycle. In this paper we present two new algorithms to enumerate all chordless cycles in a graph that have the advantage of finding each chordless cycle only once ensured by the introduction of some optimizations. The proposed algorithms are based on a recursive depth-first search strategy. We improve the first presented algorithm by the insertion of breadth-first search (BFS) strategy. This ensure that all search in a chordless path to find a new chordless cycle. The proposed algorithm using BFS has time complexity $\mathcal{O}(n + m)$ in the output size.
In a finite undirected simple graph, a {\it chordless cycle} is an induced subgraph which is a cycle. We propose two algorithms to enumerate all chordless cycles of such a graph. Compared to other similar algorithms, the proposed algorithms have the advantage of finding each chordless cycle only once. To ensure this, we introduced the concepts of vertex labeling and initial valid vertex triplet. To guarantee that the expansion of a given chordless path will always lead to a chordless cycle, we use a breadth-first search in a subgraph obtained by the elimination of many of the vertices from the original graph. The resulting algorithm has time complexity $\mathcal{O}(n + m)$ in the output size, where $n$ is the number of vertices and $m$ is the number of edges.

\end{abstract}

%%%%%%%%%%%%%%%%%%%%%%%%%%%%%%%%%%%%%%%%%%%%%%%%%%%%%%%%%%%%%%%%%%%

\section{Introduction}
\label{introduction}

Given a finite undirected simple graph $G$, a {\it chordless cycle} is an induced subgraph that is a cycle. That is, a closed sequence of vertices in $G$ such that each two adjacent vertices in the sequence are connected by an edge in $G$ and each two non-adjacent vertices in the sequence are not connected by any edge in $G$. A chordless cycle with four of more edges in termed a {\it hole}.

A solution to the problem of determining whether or not a graph contains a chordless cycle with $k \geq 4$ vertices or more, for some fixed value of $k$, was proposed by Hayward~\cite{H1987}. Golumbic~\cite{G1980} proposed an algorithm to recognize chordal graphs, that is graphs without any chordless cycles. The case for $k \geq 5$ was settled by Nikolopoulos and Palios~\cite{NP2007}. However, finding any chordless cycle with a given length $k$ is easier than finding all chordless cycles in a graph $G$.

Enumeration is a fundamental task in computer science and many algorithms have been proposed for enumerating graph structures such as cycles~\cite{DK1995,FGMPRS2012,LT1982,LW2006,RT1975,SS2007,W2008}, circuits~\cite{B2010,T1973}, paths~\cite{HH2006,RT1975}, trees~\cite{KR2000,RT1975} and cliques~\cite{MU2004,TTT2006}. Due to the number of cycles -- which can be exponentially large -- these kind of tasks are usually hard to deal with, since even a small graph may contain a huge number of such structures. Nevertheless, enumeration is necessary in many practical problems. For example, cycle enumeration is useful for the analysis of the World Wide Web and social networks, where the number of cycles can be used to identify connectivity patterns in a network.

Pfaltz~\cite{P2013} showed that chordless cycles effectively characterize connectivity structures of networks as a whole. Chordless cycles are used to better understand ecological networks structures, such as food webs, where goal is to discover the predators that compete for the same prey \cite{SBBHN2012}. To achieve this aim, the directed graph of a food web is transformed into a niche-overlap graph to highlight the competition between species. The lack of chordless cycles in the transformed graph means that the species can be rearranged as a single hierarchy. Another application is the nature of structure-property relationships in some chemical compounds that are related to the presence of chordless cycles~\cite{G2001}.

Wild~\cite{W2008} proposed an algorithm to list all subsets of cycles of cardinality at most five, using the principle of exclusion. The algorithm can be easily adapted to list only the chordless cycles. Unfortunately, the author did not present a complexity analysis of the algorithm and it is easy to see that it has a high asymptotic time complexity.

An algorithm that enumerates all chordless cycles is described by Sokhn et al.~\cite{SBBHN2012}. The general principle of this algorithm is to use a vertex ordering and to expand paths from each vertex using a depth-first search (DFS) strategy. This approach has the disadvantage of finding twice each chordless cycle.

An algorithm to enumerate chordless cycles, with $\mathcal{O}(n + m)$ time complexity in the output size, was proposed by Uno and Satoh~\cite{U2014} and, as the algorithm of Sokhn et al.~\cite{SBBHN2012}, each chordless cycle will appears more than once in the output. Actually, each cycle will appear as many times as its length. Thus, the algorithm has $\mathcal{O}(n \cdot (n + m))$ time complexity in size of the sum of lengths of all the chordless cycles in the graph, where $n$ and $m$ are the number of vertices and edges, respectively. %Its performance was evaluated by computational experiments with random graphs. Uno and Satoh showed that the computation time is constant per chordless cycle and also that the number of chordless cycles is quite smaller than the total number of cycles for those random graphs.

We propose two algorithms to enumerate all chordless cycles of a given graph $G$, with $\mathcal{O}(n+m)$ time complexity in the output size, with the advantage of finding each chordless cycle only once. The core idea of our algorithms is to use a vertex labeling scheme, with which any arbitrary cycle can be described in a unique way. With this, we generate an initial set of vertex triplets and use a DFS strategy to find all the chordless cycles. Our approach guarantees that each chordless cycle is found only once. We would like to clarify that our method, even though based on similar ideas of those presented by Sokhn et al.~\cite{SBBHN2012}, was developed independently and runs significantly faster. 

The remainder of the paper is organized as follows: some preliminaries definitions and comments are presented in Section~\ref{preliminaries}; our algorithms are introduced in Section~\ref{sequential}; Section~\ref{tests} describes the experimental tests and results produced by the new algorithms compared to other methods; finally, in Section~\ref{conclusions} we draw our conclusions. A detailed description of the algorithm is given in the Appendix.

%%%%%%%%%%%%%%%%%%%%%%%%%%%%%%%%%%%%%%%%%%%%%%%%%%%%%%%%%%%%%%%%%%%

\section{Preliminaries}
\label{preliminaries}

Let $G = (V, E)$ be a finite undirected simple graph with vertex set $V$ and edge set $E$. Let $n=|V|$ and $m=|E|$. We denote by $Adj(x)$ the set of neighbors of a vertex $x \in V$, that is, $Adj(x) = \{y \in V \>|\> (x, y) \in E\}$, and by $Adj[x] = \{x\} \cup Adj(x)$ the closed neighborhood of vertex $x$.

A \textit{simple path} is a finite sequence of vertices $\langle v_1, v_2, \dots, v_k \rangle$ such that $(v_i, v_{i+1})\allowbreak \in E$ and $v_i\neq v_j$, for each $i = 1, \dots, k$ and all $j\neq i$, $j = 1, \dots, k-1$. A \textit{cycle} is a simple path $\langle v_1, v_2, \dots, v_k \rangle$ such that $(v_k, v_1)\in E$. We denote a cycle with $k$ vertices by $C_k$. Note that our definition of cycle does not repeat the first vertex at the end of the sequence as usually done. We decided to use this definition (with the first vertex implicitly included at the end) because it simplifies the representation of a rotated version of the cycles. Note that if $\langle v_1, v_2, \dots, v_k \rangle$ is a cycle, so also are $\langle v_i, v_{i + 1} \dots, v_k, v_1, v_2, \dots, v_{i-1} \rangle$ and $\langle v_i, v_{i - 1}, \dots, v_2, v_1, v_k, \dots, v_{i+1} \rangle$, for all $i = 1, \dots, k$. A {\it chord} of a path (resp. cycle) is an edge between two vertices of the path (cycle), that is not part of the path (cycle). A path (cycle) without chord is called a {\it chordless path (chordless cycle)}. %We may also define chordless path (chordless cycle) as a path (cycle) which is an induced subgraph.

The minimum degree among all vertices of $G$ is denoted by $\delta(G)$; the maximum degree is denoted by $\Delta(G)$; and $degree_{G}(v)$ denotes the degree of a particular vertex $v\in V$. We will denote by $G - X$ ($G - u$) the subgraph induced by the subset $V - \{X\}$, for $X \subseteq V$ ($V - \{u\}$, for $u \in V$). Similar to Chandrasekharan, Laskshmanan and Medidi \cite{CLM1993}, we give a new characterization of graphs having chordless cycles of length $k\geq 4$ in Lemma~\ref{lemma_graph_charac} below.

\begin{lemma}
\label{lemma_graph_charac}
Given $t \geq 4$, a graph $G$ has a chordless cycle $C_s$, $s \geq t$, if and only if there exists a chordless path $\langle u_1, u_2, \ldots, u_t\rangle$ such that $u_1$ and $u_t$ are in the same connected component of
\begin{equation}
G' = G - \left(\bigcup^{t-1}_{i=2} Adj[u_i] - \{u_1, u_t\}\right).
\end{equation}
\end{lemma}

\begin{proof}
(Sufficiency) Suppose a chordless path $p = \langle u_1, u_2, \ldots, u_t\rangle$ such that $u_1$ and $u_t$ are in the same connected component of $G'$. Let $q = \langle v_1, v_2, \ldots, v_k\rangle$ be a shortest path (with regards to the number of edges) in $G'$ such that $v_1 = u_1$ and $v_k = u_t$. This situation is shown in Figure~\ref{tikz:chordless1}.

Suppose, by contradiction, that $p \cup q = \langle v_1 = u_1, u_2, \ldots, u_t = v_k, v_{k-1}, \ldots, v_2\rangle$ is not a chordless cycle. Thus, there exists a chord. Since $p$ is a chordless path and vertices of $p$ are vertices of $G'$, the chord must be of the form $(v_i, v_j)$ for some $i, j \in \{1, \ldots, k\}$, $i \neq j$. Moreover, we can assume that $i \leq j-2$. This leads to a path $\langle v_1, \ldots, v_i, v_j, \ldots, v_k\rangle$. Obviously, this path is also a path in $G'$ which is shorter than $q$. This yields the desired contraction. Therefore, $p \cup q$ is a chordless cycle of length $t+k-2$. Since $k \geq 2$, we have that $s = t+k-2 \geq t$.

(Necessity) Suppose there exists a chordless cycle $C_s = \langle u_1, u_2, \ldots, u_s\rangle$ in $G$ such that $s \geq t$. The chordless path $\langle u_1, u_2, \ldots, u_t\rangle$ clearly satisfies the required condition.\qed
\end{proof}

\begin{figure}
 \centering
 \begin{tikzpicture}[
  every node/.style={draw, 
                                 thick,
                                 circle,
                                 },
  scale=1.4,
  transform shape
]
  \node[label={[yshift=-5pt]above:{\scriptsize $u_2$}}] (Y1) {};
  \node[label={[yshift=-5pt]above:{\scriptsize $u_3$}},right of= Y1] (Y2) {};
  \node[label={[yshift=-10pt]above:{\scriptsize $u_{t-2}$}},right of= Y2,xshift=20pt] (YK5) {};
  \node[label={[yshift=-10pt]above:{\scriptsize $u_{t-1}$}},right of= YK5] (YK4) {};
  \node[label={[yshift=+5pt]below:{\scriptsize $u_1$}},below of= Y1] (X1) {};
  \node[label={[yshift=+5pt]below:{\scriptsize $v_2$}},right of= X1] (Z1) {};
  \node[label={[yshift=+10pt]below:{\scriptsize $v_{k-1}$}},right of= Z1,xshift=20pt] (ZQ) {};
  \node[label={[yshift=+5pt]below:{\scriptsize $u_t$}},right of= ZQ] (YK3) {};
  \draw[thick] (Y1.east) to (Y2.west);
  \draw[thick,dashed] (Y2.east) to (YK5.west);
  \draw[thick] (YK5.east) to (YK4.west);
  \draw[thick] (YK4.south) to (YK3.north);
  \draw[thick] (YK3.west) to (ZQ.east);
  \draw[thick,dashed] (ZQ.west) to (Z1.east);
  \draw[thick] (Z1.west) to (X1.east);
  \draw[thick] (X1.north) to (Y1.south);
 \end{tikzpicture}
\label{tikz:chordless1}
\caption{Sufficiency condition for the existence of a chordless cycle.}
\end{figure}
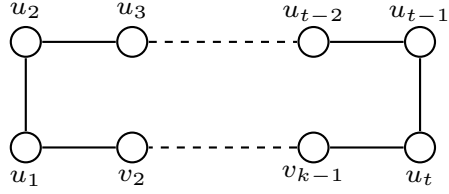

An ordering of the vertices of $G$ can be defined by a bijection $\ell: V \to \{1, 2, \dots, n\}$.  We call such a bijection a {\it vertex labeling}. Lemma~\ref{lemma_cycle} below states that any labeling enables a cycle to be defined in a unique way.

\begin{lemma}
\label{lemma_cycle}
Let $G$ be an undirected graph and $\ell\!:\!V \to \{1, 2, \dots, n\}$ a vertex labeling. If
 \begin{enumerate}[label=(\roman*), ref=(\roman*),leftmargin=2\parindent]
  \item\label{c-a} $G$ contains a simple cycle $\langle v_1, v_2, \dots, v_k \rangle$,
  \item\label{c-b} $\ell(v_2) =\allowbreak \min\{\ell(v_i) \>|\> i= 1, \dots, k\}$ and
  \item\label{c-c} $\ell(v_1) < \ell(v_3)$.
 \end{enumerate}
then $\ell$ defines the cycle in a unique way.
\end{lemma}

\begin{proof}
Any cycle $\langle v_1, v_2, \dots, v_k \rangle$ can be described as $\langle v_i, \dots,\allowbreak v_k, v_1, v_2,\allowbreak  \dots, v_{i-1} \rangle$ or $\langle v_i, \dots, \allowbreak v_2, v_1, v_k, \dots, v_{i+1} \rangle$, for all $i = 1, \dots, k$. Let $i$ be a vertex index such that $\ell(v_{i}) = \min\{\ell(v_j)$ $|\> j=1, \dots, k\}$. There are only two possibilities for the vertex $v_i$ to be the second one of the cycle: $\langle v_{i-1},v_{i},v_{i+1}, \dots, v_k, v_1, v_2, \dots, v_{i-2} \rangle\>$ or $\>\langle v_{i+1},v_{i},v_{i-1}, \dots, v_2, v_1, v_k, \dots, v_{i+2} \rangle$. Since the neighbors of $v_i$ in the cycle are $v_{i-1}$ and $v_{i+1}$, exactly one of these possibilities satisfies the condition~\ref{c-c}.\qed
\end{proof}

We define a {\it triplet} as a sequence of vertices that can initiate a possible chordless cycle of length greater than three. Let $T(G)$ denote the set of all initial valid triplets of $G$, that is, $T(G) = \{ \langle x, u, y \rangle \mid x, u, y \in V \mbox{ with } x, y \in Adj(u)$, $\ell(u) < \ell(x) < \ell(y)$ and $(x, y) \notin E\}$.

The vertex labeling that we use, named {\it degree labeling}, is constructed over a sequence of subgraphs of $G$. We start with $G_1 = G$. For $i \geq 1$, the $(i+1)^{\mbox{th}}$ subgraph is defined as $G_{i+1} = G_{i} - u_{i}$, for a chosen $u_i \in V(G_i)$ such that $degree_{G_{i}}(u_{i}) = \delta(G_{i})$. Given such a sequence, we define the degree labeling as $\ell(u_i) = i$ for each $i$. Observe that  for any chosen labeling, if $G$ is a tree there are no possible triplets, that is, $T(G) = \varnothing$. Moreover, if $G$ has a unique cycle then $|T(G)| = 1$, no matter what degree labeling is used, that is, unneeded triplets are discarded.

Lemma~\ref{lemma_paths} below establishes the possible properties of a neighbor of the last vertex of a chordless path.

\begin{lemma} \label{lemma_paths}
Let $p = \langle v_1, v_2, \dots, v_k \rangle$ be a chordless path and $v \in Adj(v_k)$, $v \neq v_{k-1}$. Exactly one of the following occurs:
 \begin{enumerate}[label=(\roman*), ref=(\roman*),leftmargin=2\parindent]
   \item\label{p-a} $\langle p, v  \rangle = \langle v_1, v_2,\allowbreak \dots, v_k, v \rangle$ is a chordless path;
   \item\label{p-b}  $(v_{k-1},v) \in E$; or
   \item\label{p-c}  there exists $i\in\{1, \dots,\allowbreak k-2\}$ such that $p = \langle v_i, v_{i+1}, \dots, v_k, v \rangle$ is a chordless cycle.
 \end{enumerate}
\end{lemma}

\begin{proof}
Since $v \in Adj(v_k)$, $v \neq v_{k-1}$ and $p$ is a chordless path, then $\langle p, v\rangle$ is a simple path. Suppose that $(v_{k-1},v ) \notin E$ and that $\langle p, v\rangle$ is not a chordless path. Therefore, there is an index $i \in \{1, \ldots, k-2\}$ with $(v, v_i) \in E$. Choosing the biggest index $i$ with this property, we have the desired chordless cycle.\qed
\end{proof}

Case~\ref{p-a} states that path $\langle p, v\rangle$ can be part of a chordless cycle. Cases~\ref{p-b} and~\ref{p-c}, with $i \neq 1$, state that path $\langle p, v\rangle$ has a chord. In case~\ref{p-c} with $i=1$, $\langle p, v\rangle$ is a chordless cycle.

%%%%%%%%%%%%%%%%%%%%%%%%%%%%%%%%%%%%%%%%%%%%%%%%%%%%%%%%%%%%%%%%%%%

\section{Algorithms to find all chordless cycles}
\label{sequential}

The general principle of the proposed algorithms is to limit the search space by creating an initial set of valid vertex triplets ($T(G)$) and use a DFS strategy to identify chordless paths from each triplet in $T(G)$. At each step our strategy employs several optimization techniques to reduce the number of required tests before deciding what to do with the current expanded path $\langle p \rangle$. Among these techniques, the most relevant are blocking and labeling. The first one is an adaptation of a technique originally presented by Tiernan \cite{T1970} and subsequently used by Tarjan \cite{T1972}, Johnson \cite{J1975} and Read and Tarjan \cite{RT1975}. All these techniques are described in detail in~Appendix \ref{appendix-a}. The proposed algorithm is presented next.

The algorithm \textit{ChordlessCycles(G)} works as follows. It starts constructing a vertex labeling $\ell:V\to\{1,2,\dots, n\}$ and a set $T(G)$, of all initial valid vertex triplets (step~\ref{cc-step-1}). The set $T$ contains paths still open to expansion and is, at the start, comprises the initial valid vertex triplets (step~\ref{cc-step-2}).

The set $C$ is initialized (step~\ref{cc-step-3}) with all triangles (which are all chordless).

%-----------------------------------------------------------------------------------------------------------------

\subsection{Algorithm overview}
\label{algo-overview}

\begin{center}
\begin{minipage}{.95\textwidth}
\begin{algorithm2e}[H]
   \SetKwFunction{ccvisit}{CC\_Visit}
   \LinesNumbered

   \KwIn{Graph $G$.}
   \KwOut{Set $C$ of all chordless cycles of $G$.}
   \BlankLine

   $T(G) \leftarrow \{ \langle x, u, y \rangle \mid x, u, y \in V \mbox{ with } x, y \in Adj(u)$, $\ell(u) < \ell(x) < \ell(y)$ and $(x, y) \notin E\}$\nllabel{cc-step-1}\;
   \BlankLine

   $T \leftarrow T(G)$\nllabel{cc-step-2}\;
   $C \leftarrow \{ \langle x, u, y \rangle \mid x, u, y \in V \mbox{ with } x, y \in Adj(u)$, $\ell(u) < \ell(x) < \ell(y)$ and $(x, y) \in E\}$\nllabel{cc-step-3}\;
   \BlankLine

   \While{$T \neq \varnothing$}{\nllabel{cc-step-4}
      $p \leftarrow \langle u_1, u_2, \ldots, u_t \rangle \in T$\nllabel{cc-step-5}\;
      $T \leftarrow T - \{p\}$\nllabel{cc-step-6}\;
   \BlankLine

   \ForEach{$v \in Adj(u_t)$}{\nllabel{cc-step-7}
	\If{$((\ell(v) > \ell(u_2))\ \mathbf{and}\ (v \notin Adj(u_i),\> i \in \{2, \dots, t-1\}))$\nllabel{cc-step-8}}{
	  \eIf{$v \in Adj(u_1)$\nllabel{cc-step-9}}{
	    $C \leftarrow C \cup \{\langle p, v \rangle\}$\nllabel{cc-step-10}\;
	  }
	  {
	    $T \leftarrow T \cup \{\langle p, v \rangle\}$\nllabel{cc-step-11}\;
	  }
	}
       }
      \BlankLine
    }
   \BlankLine
   \Return{$C$.}

   \caption{\textit{ChordlessCycles(G)}
   \label{alg:cicloSemCordaAltoNivel}}
\end{algorithm2e}
\end{minipage}
\end{center}

%-----------------------------------------------------------------------------------------------------------------

Each initial triplet $\langle x, u, y \rangle\in T(G)$ is expanded using the neighbors of the vertex $y$ (steps \ref{cc-step-4}--\ref{cc-step-11}). The algorithm checks (step \ref{cc-step-8}) if the addition of a neighbor of $y$ to the path gives:
 \begin{enumerate}[label=(\roman*), ref=(\roman*),leftmargin=2\parindent]
   \item\label{case1} a chordless cycle;
   \item\label{case2} a chord in the current path; or
   \item\label{case3} another expansible path.
 \end{enumerate}

In case~\ref{case1}, the newly chordless cycle found is added to the set of cycles (step \ref{cc-step-10}); in case~\ref{case2}, the path is discarded and, in the last case, the expanded path is added to the set $T$ of expandable paths (step \ref{cc-step-11}). The same process is repeated until the set $T$ becames empty. An in-depth version of the algorithm \textit{ChordlessCycles(G)} is given in Appendix~\ref{appendix-a}.

In the implementation, we use a version of blocking up to facilitate the condition in step~\ref{cc-step-8}. As stated before, the idea of blocking was originally proposed by Tiernan \cite{T1970}. However, his strategy, consisting of the simple blocking and unblocking of vertices, is not enough to meet the needs of the proposed algorithms, because each vertex in the chordless path can be a neighbor of several others. Thus, we expand the concept and use a counter that indicates the number of times a vertex is found as a neighbor of some other in a chordless path. The idea here is that the current vertex can be considered unblocked only when all of its neighbors have already been processed.

For each $v\in V$, let $blocked[v]$ denote the number of neighbors of $v$ in a chordless path (without the first vertex). A vertex $v$ is said to be {\it unblocked} if $blocked[v] = 0$ and {\it blocked} otherwise.  At the beginning of the processing of a triplet, except for its first vertex, the neighbors of the other two vertices are marked as blocked and can no longer be used in the path. The goal of this approach is to block the neighbors that could form a chord with vertices in the chordless path. This strategy enables the extension of the path in a faster way. Upon completion of the processing of a triplet, all these neighbors are marked as unblocked. The blocking and unblocking operations are detailed in Algorithms~\ref{alg:block_neighbors} (\textit{BlockNeighbors()}) and~\ref{alg:unblock_neighbors} (\textit{UnblockNeighbors()}), respectively.

Algorithm~\ref{alg:cc_visit} (\textit{CC\_Visit()}) extends a chordless path from the last vertex of the triplet using a DFS strategy. At each step of the recursion, it checks what will be the result of the addition of each neighbor of the last vertex in the current path. If the addition results in a chordless cycle, it is added to the set of cycles already found and the recursion ends. If the addition results in an expanded chordless path, it is added to the set of paths to be analyzed. If it forms a chord, the path is discarded.

To ensure that each performed search finds a chordless cycle, the steps \ref{cc-step-7}--\ref{cc-step-8} are modified in order to use a breadth-first search (BFS) on the subgraph induced by the removal of vertices from the current path. We also discard all vertices $v \in V$ such that $\ell(v) \geq \ell(u_2)$. By using BFS we can verify in time $\mathcal{O}(n + m)$ that two vertices $u$ and $v$ belong to the same connected component. In this case, any chordless path can be extended to a chordless cycle.

The following algorithm extension demonstrates this modification, where $\pi(v)$ denotes the predecessor of vertex $v$ in the path of the generated search tree by the execution of BFS. Thus, given the path $\langle u_1, u_2, \ldots, u_t \rangle$ the expansion will be performed only if $u_t$ is a descendant of $u_1$ in the search tree (this is characterized by the existence of $\pi(u_t))$. In this case, a vertex $v$ adjacent to $u_t$ will be considered to expand this path only if $v$ is a descendant of $u_1$ in the search tree and $\ell(v)> \ell(u_2)$, which is ensured by Lemma \ref{lemma_graph_charac}.

Algorithms {\it ChordlessCycles(G)} [\ref{alg:chordless_cycles}] and {\it CC-Visit(G)} [\ref{alg:new_labeling}] can be easily changed so that the BFS is used, as shown in Algorithm \ref{alg:cicloSemCordaAltoNivel}.

%-----------------------------------------------------------------------------------------------------------------
\begin{center}
\begin{minipage}{.85\textwidth}
\setlength{\interspacetitleruled}{0pt}%
\setlength{\algotitleheightrule}{0pt}%
\begin{algorithm2e}[H]
      \NoCaptionOfAlgo
      BFS($u_1, G - \left(\bigcup\limits^{t-1}_{i=2}Adj[u_i] - \{u_1, u_t\} \right) - \{v \mid \ell(v) < \ell(u_2)\}$) \nllabel{cc2-step-1}\;
      \BlankLine

      \If{$(\exists \pi(u_t))$\nllabel{cc2-step-2}}{
	\ForEach{$v \in Adj(u_t)$}{\nllabel{cc2-step-3}
	  \If{$((\exists \pi(v))\ \mathbf{and}\ (\ell(v) >\ell(u_2)))$\nllabel{cc2-step-4}}{}}}

\end{algorithm2e}
\end{minipage}
\end{center}
%-----------------------------------------------------------------------------------------------------------------

Recall that by Lemma \ref{lemma_graph_charac}, if there exists a chordless path $p = \langle u_1, u_2, \ldots, u_t\rangle$ and a shortest path $q = \langle v_1, v_2, \ldots, v_k\rangle$ in the induced subgraph $G' = G - \left(\bigcup\limits^{t-1}_{i=2}Adj[u_i] - \{u_1, u_t\} \right)$ between $u_1$ and $u_t$, then $p \cup q$ forms a chordless cycle.

\subsection{Algorithm correctness}

The correctness of the algorithm is due to the fact that no vertex is kept blocked at the end of its execution, which is guaranteed by Lemma~\ref{lemma_block_CC-Visit} below.

\begin{lemma}
\label{lemma_block_CC-Visit}
Let $p = \langle u_1, u_2, \dots, u_t\rangle$ be a chordless path.
At the beginning and at the end of each execution of the algorithm \textit{CC\_Visit()}, $blocked[v] = k$ if and only if $v$ is a neighbor of $k$ vertices in $\{u_2, \dots, u_{t-1}\}$, for any vertex $v \in V$.
\end{lemma}

\begin{proof}
As in the proof of Lemma~\ref{lemma_paths}, it is not hard to see that at beginning of any CC\_Visit($p, C,$ $key$, $blocked$) execution, we have increased the counter $blocked[v]$ by 1 for all neighbors of each vertex in $\{u_2, \dots, u_{t-1}\}$; so, at the end of the process, we have also decreased each one of these values, ensuring that $blocked[v] = k$.\qed
\end{proof}

Let $\langle u_1, \dots, u_t\rangle$ be a chordless path and $v \in Adj(u_t)$. From Lemmas~\ref{lemma_paths} and~\ref{lemma_block_CC-Visit} it is easy to see that, after the call of algorithm \textit{BlockNeighbors()} to block the $u_t$ neighbors, $blocked[v] = 1$ if and only if $\langle p, v \rangle$ is a chordless path or a chordless cycle.

Using Lemmas~\ref{lemma_cycle}, \ref{lemma_paths} and \ref{lemma_block_CC-Visit} we now demonstrate the correctness of the proposed algorithm, as stated by the following theorem.

\begin{theorem}{\bf Correctness of the algorithm \textit{ChordlessCycles(G)}}.\\
The algorithm \textit{ChordlessCycles(G)} enumerates all the chordless cycles of a graph $G$.
\end{theorem}

\begin{proof}
Let $C=\langle u_1, u_2, u_3, \dots, u_k, u_1\rangle$ be a chordless cycle of $G$. By Lemma~\ref{lemma_cycle} we can assume that $\ell(u_2) = \min\{\ell(u_i) \>|\> i=1, \dots, k\}$ and $\ell(u_1) < \ell(u_3)$. Therefore, the triplet $\langle u_1, u_2, u_3 \rangle$ is generated by the algorithm \textit{Triplets(G)}. Thus, the algorithm \textit{ChordlessCycles(G)} performs steps~\ref{cc-step-7} to~\ref{cc-step-11} with $p = \langle u_1, u_2, u_3 \rangle$.

Now let $\langle v_1, \dots, v_s = u_1, \dots, u_i\rangle$ be a chordless path and $v \in Adj(u_i)$. Combining Lemmas~\ref{lemma_paths} and~\ref{lemma_block_CC-Visit}, after the call to algorithm \textit{BlockNeighbors()} in \textit{CC\_Visit()} to block the neighbors of the vertex $u_i$, $blocked[v] = 1$ if and only if $\langle p, v \rangle$ is a chordless path or a chordless cycle. In the first case, the \textit{CC\_Visit()} will be called again until eventually $i=k$, finding the cycle; in the second case we already have the desired chordless cycle and it is added to the set $C$.\qed
\end{proof}

\subsection{Algorithm complexity}

Note that the depth of each search is at most the length of the longest chordless path. Moreover, the number of calls to BFS is limited by the output size. Since BFS is performed in $\mathcal{O}(n + m)$ time complexity, our algorithm has time complexity of $\mathcal{O}(n+m)$ in the output size. Actually, the BFS operates on a subgraph that diminishes with each iteration. The best algorithm of which we are aware to find all chordless cycles in a graph $G$~\cite{U2014} has time $\mathcal{O}(n \cdot (n + m))$ in the output size since it finds the same chordless cycle more than once.

In order to reduce the algorithm execution time, the biconnected components identification strategies presented by Tarjan~\cite{T1972} and Szwarcfiter~\cite{S1988}, that have $\mathcal{O}(n^2)$ time complexity, could be used. These strategies discard all vertices with $\delta(G) < 2$ and some paths that cannot lead to a chordless cycle.

\subsection{Algorithm improvement}

To improve the overall process even further, it is possible to preprocess each initial valid triplet $\langle x, u, y \rangle\in T(G)$. Recall that algorithm \textit{ChordlessCycles(G)} extends a chordless path only from the $y$ extremity. So, the preprocessing step could be used to try to extend the chordless path from the $x$ extremity.  This extension may be performed whenever there is only one vertex near to the extremity of the path, an unblocked neighbor, that can be added to it. Thus, the chordless path would be extended in a unique way and would not change the number of examined chordless paths. Moreover, this extension may reduce the work of the algorithm \textit{ChordlessCycles(G)}, since it blocks the neighbors of the visited vertices and they do not need to be examined anymore.

The blockings made from the $x$ extremity, after all possible extensions, prevent extensions to be made from the $y$ extremity to use any vertex already inserted in the opposite end, unless it forms a chordless cycle. They also avoid the usage of vertices that have been blocked, that would form chords on the path.

After a triplet processing, all blocked vertices in the extensions from $x$ and $y$ will be unblocked. Lemma \ref{lemma_block_PrimeExtend}, below, ensures that no vertex blocked from the $x$ extremity remains in this state after the processing of a triplet.

\begin{lemma}
\label{lemma_block_PrimeExtend}
Let $p = \langle x, u, y\rangle$ and $q = \langle v_1, v_2, \dots, (v_s = x) \rangle$ be paths in $G$, such that $\langle q, u, y \rangle$ is a chordless path. At the beginning of each execution of the extension step, for any vertex $v \in V$, $blocked[v] = k$ if and only if $v$ is a neighbor of $k$ vertices in $\{v_2, \dots, (v_s=x), u, y\}$.
\end{lemma}

\begin{proof}
At the first execution of an extension step, the counter $blocked[v]$ is increased by 1 unit for all neighbors of $u$ and $y$.  Its next execution is performed after the path $q$ is augmented to $\langle z, q \rangle$, for some $z\in Adj(v_1)$, and the counter $blocked[v]$ is increased by 1 for all neighbors of $v_1$. Thus, the result holds for any execution of the extension step since the counter $blocked[v]$ is increased by 1 unit for all neighbors of each vertex in $\{v_2, \dots, (v_s=x), u, y\}$. \qed
\end{proof}

Let $\langle q, u, y \rangle = \langle v_1, v_2, \dots, (v_s\allowbreak = x), u, y \rangle$ be a chordless path and $v \in Adj(v_1)$. Combining Lemma~\ref{lemma_paths} with Lemma~\ref{lemma_block_PrimeExtend}, it is easy to see that $blocked[v] = 0$ if and only if $\langle v, q, u, y \rangle$ is a chordless path and $blocked[v] = 1$ if and only if $\langle x, u, y, v, q\rangle$ is a chordless cycle, for $(v,y) \in E$. In the last case, the cycle $\langle x, u, y, v, q\rangle$ is equivalent to the cycle $\langle v, q, u, y, v \rangle$. Furthermore, Lemma~\ref{lemma_block_PrimeExtend} states that no vertex will be kept blocked after the extension from the $x$ extremity, while Lemma~\ref{lemma_block_CC-Visit} guarantees that the same occurs when starting at the $y$ extremity.

%%%%%%%%%%%%%%%%%%%%%%%%%%%%%%%%%%%%%%%%%%%%%%%%%%%%%%%%%%%%%%%%%%%

\section{Experimental results}
\label{tests}

In the implementation of the algorithm $ChordlessCycles(G)$ we used two data structures: an adjacency matrix, that enables the verification of adjacency between two vertices in constant time, and also a compact representation of graphs as proposed by Harish and Narayanan~\cite{HN2007}.

The implementation and execution of the algorithm was performed using C++ and the g++ compiler on Linux openSUSE 12.3 ``Dartmouth'' operating system, a HP Proliant DL380 G7 Xeon Quad Core E5506 2.13GHZ with 40GB of RAM memory and 1.6 TB of disc.

The running time (in seconds) for the graphs presented in~\cite{SBBHN2012}, representing some ecological networks, named {\it food web}, and also other well-known graphs are shown in Table~\ref{tab:Tempo_Exec1}. The column labeled ``Name'' is the dataset name, $n$ is the number of vertices, $m$ is the number of edges, $\# clc$ is the number of chordless cycles of length of four or more and $C_3$ is the number of cycles of length 3 in the graph. $T_1$, $T_2$, $T_3$ and $T_4$ represent, respectively, the running time (in seconds) of the algorithm of Sokhn et al.~\cite{SBBHN2012} as presented originally by the authors and also when executed on our machine, of our algorithm and of its modified version with BFS. The symbol ``--'' in column $T_1$ refers to untested graphs by Sokhn et al.~\cite{SBBHN2012}. Therefore, we just present the results of their algorithm when it runs on our computer.

In Table~\ref{tab:More_information}, $l-clp$, $\#vis$ and $\#rec$ refer, respectively, to the longest chordless path of the graph,  the vertex visit quantity and the recursion carried out in a search for a chordless cycles.
 
\begin{table}[htpb]
\centering
\caption{Running time to enumerate all chordless cycles on niche-overlap graphs and on other well-known graphs.}
\label{tab:Tempo_Exec1}
\begin{tabular}{llcccccccc}
\hline
 \textbf{Id.}	& \textbf{Name} 		& \textbf{$n$} 		& \textbf{$m$} 	& \textbf{\#clc} 	& $C_3$	& \textbf{ $T_1$} 	& \textbf{ $T_2$} 	& \textbf{ $T_3$} & \textbf{ $T_4$}\\
\hline
1	& CrystalD 		& 16			& 86		& 0			& 293	& --			& 0.00			& 0.00	& 0.00 \\
2	& ChesUpper		& 24			& 85		& 0			& 167	& --			& 0.00			& 0.00	& 0.00  \\
3	& Narragan			& 26			& 168		& 0			& 586	& --			& 0.00			& 0.00	& 0.00 \\
4	& Chesapeake		& 27			& 90		& 0			& 157	& 0.00			& 0.00			& 0.00	& 0.00 \\
5	& Michigan			& 29			& 175		& 0			& 587	& --			& 0.00			& 0.00	& 0.00 \\
6	& Mondego			& 30			& 206		& 0			& 886	& --			& 0.00			& 0.00	& 0.00 \\
7	& Cypwet			& 53			& 842		& 0			& 8946	& 6.00			& 0.01			& 0.01	& 0.01 \\
8	& Everglades		& 58			& 1214		& 710			& 15627	& 7.00			& 0.03			& 0.03	& 0.04 \\
9	& Mangrovedry		& 86			& 2132		& 27426			& 30659	& 359.00			& 1.10			& 0.31	& 0.78 \\
10	& Floridabay		& 107			& 3249		& 85976			& 62389	& 4569.00			& 11.57		& 1.03	& 5.37 \\
\hline
11	& Goi\^ania		& 43		& 75		& 9311	& 5		& -- & 0.54			& 0.11	& 0.19 \\
12	& $C_{100}$		& 100	& 100	& 1		& 0		& -- & 0.00			& 0.00	& 0.00 \\
13	& Wheel 100		& 101	& 200	& 1		& 100	& -- & 0.00			& 0.00	& 0.00 \\
14	& $K_{8, 8}$		& 16		& 64		& 784	& 0		& -- & 0.00	& 0.00	& 0.00 \\
15	& $K_{50, 50}$	& 100	& 2500	& 1500625		& 0	& -- & 1.17			& 1.58	& 2.88 \\
16	& Grid 4 $\times$ 10		& 40		& 66		& 1823		& 0	& -- & 0.13			& 0.03	& 0.05 \\
17	& Grid 5 $\times$ 6		& 30		& 49		& 749		& 0	& -- & 0.01			& 0.00	& 0.01 \\
18	& Grid 5 $\times$ 10		& 50		& 85		& 52620		& 0	& -- & 2.20			& 0.60	& 1.13 \\
19	& Grid 6 $\times$ 6		& 36		& 60		& 3436		& 0	& -- & 0.07		& 0.02	& 0.06 \\
20	& Grid 6 $\times$ 10		& 60		& 104	& 800139		& 0	& -- & 37.79		& 9.15	& 16.83 \\
21	& Grid 7 $\times$ 10		& 70		& 123	& 8136453	& 0	& -- & 678.09		& 85.23	& 189.86 \\
\hline
\end{tabular}
\end{table}

\begin{table}[htpb]
\centering
\caption{More information about the algorithm execution.}
\label{tab:More_information}
\begin{tabular}{lccccp{1pt}cc}
 \hline
	& & & \multicolumn{2}{c}{\bf Algorithm without BFS}	&& \multicolumn{2}{c}{\bf Algorithm using BFS}\\
 \cline{4-5}\cline{7-8}\\[-8pt]
 \textbf{Id.}	& \textbf{$|T(G)|$}	& \textbf{l-clp}	& \textbf{\#vis}	& \textbf{\#rec}	&& \textbf{\#vis}	& \textbf{\#rec}\\
\hline
 1 	& 16		& 3	& 1 205		& 16 	&& 1 293 	& 16 \\
 2	& 31		& 3	& 7 119		& 307	&& 2 023 	& 31 \\
 3	& 59		& 3	& 9 646		& 215	&& 6 062	& 59 \\
 4	& 21		& 3	& 1 057		& 27	 	&& 1 102 	& 21 \\
 5	& 92		& 3	& 14 212		& 257	&& 9 563	& 92 \\
 6	& 80		& 3	& 35 582		& 826	&& 8 776	& 80 \\
 7	& 909	& 3	& 244 559		& 2 473	&& 193 457	& 909 \\
 8	& 1 877	& 6	& 1 008 576	& 10 283	&& 755 141	& 2 410\\
 9	& 4 095	& 8	& 23 211 495	& 226 078	&& 15 984 224 & 42 157\\
 10	& 5 837	& 8	& 108 212 550	& 685 492	&& 107 989 118	& 273 130 \\
 \hline
 11	& 31	 	& 28	& 1 298 846	& 128 623 && 1 278 623 & 36 785 \\
 12	& 1		& 100	& 587	& 97		&& 10 287 	& 97 \\
 13	& 1		& 100	& 980	& 97		&& 15 530 	& 97 \\
 14	& 140	& 4		& 5 740	& 140	&& 13 132 		& 140 \\
 15	& 61 250	& 4	& 15 373 750	& 61 250	&& 93 467 500 & 61 250\\
 16	& 27		& 28		& 412 944	& 44 846	&& 308 928 	& 9 648\\
 17	& 20		& 18		& 34 896	& 3 739	&& 75 669		& 2 365\\
 18	& 36		& 32		& 7 161 919& 729 706&& 7 334 301	& 225 960 \\
 19	& 25		& 20		& 178 672	& 18 671	&	& 362 840 & 10 833 \\
 20	& 45		& 36		& 108 866 318	& 10 696 603	&& 105 704 459 & 3 088 973 \\
 21	& 54		& 44		& 1 447 348 446	& 140 095 162	&& 1 128 865 130 & 30 945 512 \\
\hline
\end{tabular}
\end{table}

%\texttt{
%\begin{figure}[htpb]
%  \centering
%  \includegraphics[width = .7\textwidth]{Goiania_Downtown.eps}
%  \caption{Simple representation of Goi\^ania downtown, Goi\'as, Brazil.}
%  \label{fig:Goiania_Downtown}
%\end{figure}
%}

\begin{figure}[htpb]
\begin{tikzpicture}[
  dotstyle/.style={ column sep=.5cm,
                            row sep=.35cm,
                            nodes={circle,draw,fill=white}
                          },
  line1/.style={draw,black,very thin},
  line2/.style={draw,black,line width=3pt},
]
  \matrix (m) [matrix of nodes,ampersand replacement=\&,dotstyle]
  {
% 1 \& 2 \& 3 \& 4 \& 5  \& 6 \& 7  \& 8 \& 9  \& 10 \& 11 \& 12 \\
    {} \&    \&    \& {} \& {} \& {} \& {} \& {} \& {} \& {}   \& {}  \&      \\
       \& {} \&    \& {} \& {} \& {} \& {} \& {} \& {} \& {}   \&     \& {}   \\
       \&    \& {} \&     \& {} \& {} \& {} \& {} \& {} \&      \& {}  \&      \\
       \&    \&    \& {}  \& {} \& {} \& {} \& {} \& {} \& {}   \&     \&      \\
       \&    \&    \&     \& {} \& {} \& {} \& {} \& {} \&       \&     \&      \\
       \&    \&    \&     \&    \& {} \& {} \& {} \&     \&       \&     \&      \\%[-15pt]
       \&    \&    \&     \&    \& {} \& {} \& {} \&     \&       \&     \&      \\
  };

  \draw [line1] (m-1-1) -- (m-1-4);
  \draw [line1] (m-1-4) -- (m-1-5);
  \draw [line1] (m-1-5) -- (m-1-6);
  \draw [line1] (m-1-6) -- (m-1-7);
  \draw [line1] (m-1-7) -- (m-1-8);
  \draw [line2] (m-1-8) -- (m-1-9);
  \draw [line2] (m-1-9) -- (m-1-10);
  \draw [line2] (m-1-10) -- (m-1-11);

  \draw [line1] (m-2-2) -- (m-2-4);
  \draw [line2] (m-2-4) -- (m-2-5);
  \draw [line1] (m-2-5) -- (m-2-6);
  \draw [line1] (m-2-6) -- (m-2-7);
  \draw [line1] (m-2-7) -- (m-2-8);
  \draw [line1] (m-2-8) -- (m-2-9);
  \draw [line1] (m-2-9) -- (m-2-10);
  \draw [line1] (m-2-10) -- (m-2-12);

  \draw [line2] (m-3-3) -- (m-3-5);
  \draw [line1] (m-3-5) -- (m-3-6);
  \draw [line1] (m-3-6) -- (m-3-7);
  \draw [line1] (m-3-7) -- (m-3-8);
  \draw [line1] (m-3-8) -- (m-3-9);
  \draw [line1] (m-3-9) -- (m-3-11);

  \draw [line1] (m-4-4) -- (m-4-5);
  \draw [line2] (m-4-5) -- (m-4-6);
  \draw [line2] (m-4-6) -- (m-4-7);
  \draw [line1] (m-4-7) -- (m-4-8);
  \draw [line1] (m-4-8) -- (m-4-9);
  \draw [line1] (m-4-9) -- (m-4-10);

  \draw [line1] (m-5-5) -- (m-5-6);
  \draw [line1] (m-5-6) -- (m-5-7);
  \draw [line1] (m-5-7) -- (m-5-8);
  \draw [line2] (m-5-8) -- (m-5-9);

  \draw [line2] (m-6-6) -- (m-6-7);
  \draw [line1] (m-6-7) -- (m-6-8);

  \draw [line1] (m-7-6) -- (m-7-7);
  \draw [line1] (m-7-7) -- (m-7-8);
  \draw [line1] (m-1-1) -- (m-2-2);
  \draw [line1] (m-1-4) -- (m-2-4);
  \draw [line1] (m-1-5) -- (m-2-5);
  \draw [line1] (m-1-6) -- (m-2-6);
  \draw [line1] (m-1-7) -- (m-2-7);
  \draw [line2] (m-1-8) -- (m-2-8);
  \draw [line1] (m-1-9) -- (m-2-9);
  \draw [line1] (m-1-10) -- (m-2-10);
  \draw [line2] (m-1-11) -- (m-2-12);

  \draw [line1] (m-2-2) -- (m-3-3);
  \draw [line2] (m-2-4) -- (m-3-3);
  \draw [line2] (m-2-5) -- (m-3-5);
  \draw [line1] (m-2-6) -- (m-3-6);
  \draw [line1] (m-2-7) -- (m-3-7);
  \draw [line2] (m-2-8) -- (m-3-8);
  \draw [line1] (m-2-9) -- (m-3-9);
  \draw [line2] (m-2-12) -- (m-3-11);

  \draw [line1] (m-3-3) -- (m-4-4);
  \draw [line1] (m-3-5) -- (m-4-5);
  \draw [line1] (m-3-6) -- (m-4-6);
  \draw [line1] (m-3-7) -- (m-4-7);
  \draw [line2] (m-3-8) -- (m-4-8);
  \draw [line1] (m-3-9) -- (m-4-9);
  \draw [line2] (m-3-11) -- (m-4-10);

  \draw [line1] (m-4-4) -- (m-5-5);
  \draw [line2] (m-4-5) -- (m-5-5);
  \draw [line1] (m-4-6) -- (m-5-6);
  \draw [line2] (m-4-7) -- (m-5-7);
  \draw [line2] (m-4-8) -- (m-5-8);
  \draw [line1] (m-4-9) -- (m-5-9);
  \draw [line2] (m-4-10) -- (m-5-9);

  \draw [line2] (m-5-5) -- (m-6-6);
  \draw [line1] (m-5-6) -- (m-6-6);
  \draw [line2] (m-5-7) -- (m-6-7);
  \draw [line1] (m-5-8) -- (m-6-8);
  \draw [line1] (m-5-9) -- (m-6-8);

  \draw [line1] (m-6-6) -- (m-7-6);
  \draw [line1] (m-6-7) -- (m-7-7);
  \draw [line1] (m-6-8) -- (m-7-8);

\end{tikzpicture}
  \caption{Simple representation of downtown Goi\^ania, Goi\'as, Brazil.}
  \label{fig:Goiania_Downtown}
\end{figure}
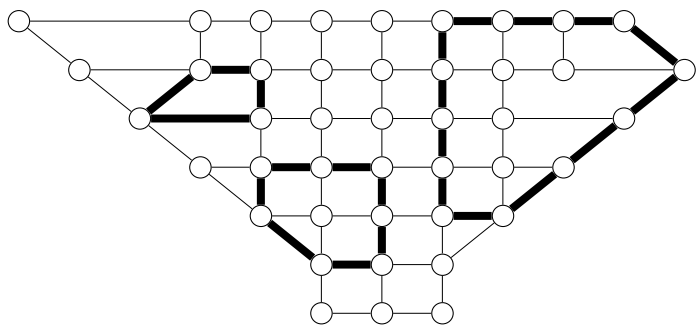

%As can observed in the Table~\ref{tab:Tempo_Exec1}, the symbol ``--'' refers to untested graphs by Sokhn et al.~\cite{SBBHN2012}; therefore, we just present the results of their algorithm running in our computer.

The first ten graphs (Ids. 1--10) were provided by known databases of ecological studies, in which the directed graph {\it food web} is transformed into undirected {\it niche-overlap} graph according to the definitions of Wilson and Watkins~\cite{WW1990}. However, in order to perform a fair comparison with the times obtained by Sokhn et al.~\cite{SBBHN2012}, we also excluded the vertices of degree 0 in the graphs of Table~\ref{tab:Tempo_Exec1}.

Figure~\ref{fig:Goiania_Downtown} shows a graph that represents the downtown of the city of Goi\^ania, the capital of state of Goi\'as, in Brazil. The running time to enumerate all chordless cycles of this graph is presented in line 11 of Table~\ref{tab:Tempo_Exec1}. Three of the 9316 chordless cycles present in this graph are highlighted in the figure.

We can observe in Table~\ref{tab:Tempo_Exec1}, the running time of our algorithm without BFS is faster than the BFS version.  Although, in principle, the search using BFS should have a smaller running time, this did not happens because BFS, which is used to verify the existence of a path between two vertices, can potentially require $\mathcal{O}(|V|) $ time complexity for every chordless path extension.

%%%%%%%%%%%%%%%%%%%%%%%%%%%%%%%%%%%%%%%%%%%%%%%%%%%%%%%%%%%%%%%%%%%

\section{Conclusions}
\label{conclusions}

%In this paper we present two new algorithms to enumerate all chordless cycles in a graph. Compared to other algorithms with the same purpose, these two new algorithms have the advantage of finding each chordless cycle only once. To ensure this, we introduce the concepts of vertex labeling and initial valid vertex triplet. We also use a improved version of the concept of vertex blocking, previously introduced by other authors. The proposed algorithms use all this techniques together with a depth-first search strategy in the graph.
%
%We improve the first presented algorithm by the insertion of breadth-first search (BFS) strategy. This ensure that all search in a chordless path to find a new chordless cycle. This search is performed in a subgraph, that is simulate by vertex blocking. The algorithm using BFS has time complexity $\mathcal{O}(n + m)$ in the output size.
%
%Despite the algorithm that uses BFS has smaller time complexity, we observe that in practice the algorithm without BFS has a faster running time.
%
%For future work, we suggest the improvement implementating the chordless path extension by other extremity of each initial valid triplet. This should make the algorithm faster.  We are implementing a parallel version of the algorithms presented.

We presented two algorithms (one with and the other without BFS) to enumerate all chordless cycles in a graph. Compared to other similar algorithms, they have the advantage of finding each chordless cycle only once. To ensure this, we introduced the concepts of vertex labeling and an initial valid vertex triplet. We also used an improved version of the concept of vertex blocking, previously introduced by other authors. The proposed algorithms use all these techniques together with a DFS strategy in a graph.

The first algorithm does not guarantee that the expansion of a given chordless path will always lead to a chordless cycle. To ensure this, we propose a second algorithm where a breadth-first search (BFS) is performed in a subgraph obtained by the elimination of all blocked vertices from the original graph. The algorithm using BFS has time complexity $\mathcal{O}(n + m)$ in the output size.

Although the algorithm that uses the BFS strategy ensures that no unsuccessful searches are performed, it is observed that, in practice, the algorithm without this technique leads to a smaller execution time.

For future work, we plan to improve the algorithms by implementing the expansion of each initial valid triplet from its two extremities. This will make the algorithms faster. Currently, a GPGPU parallel version of the algorithms described in this article is being implemented and tested.

%%%%%%%%%%%%%%%%%%%%%%%%%%%%%%%%%%%%%%%%%%%%%%%%%%%%%%%%%%%%%%%%%%%

\begin{acknowledgements}
We would like to thank Nayla Sokhn, University of Fribourg, who kindly helped us to understand the practical applications of chordless cycles in ecological networks and provided her database and algorithm for comparison. We also would like to thank Joe Jordan, Imperial College London and professors Hugo Alexandre Dantas do Nascimento and Leslie Richard Folds, from Universidade Federal de Goi\'{a}s, who gave some suggestions to improve the paper.

\end{acknowledgements}

%%%%%%%%%%%%%%%%%%%%%%%%%%%%%%%%%%%%%%%%%%%%%%%%%%%%%%%%%%%%%%%%%%%

\newpage

\appendix
\section{Detailed proposed algorithm}
\label{appendix-a}

To simplify the explanation, the algorithm is divided into six parts, denoted Algorithms~\ref{alg:chordless_cycles}--\ref{alg:unblock_neighbors}. The main part is Algorithm~\ref{alg:chordless_cycles} (\textit{ChordlessCycles(G)}). Initially, it calls Algorithm~\ref{alg:new_labeling} (\textit{DegreeLabeling()}) to generate a degree labeling for the vertices of the graph $G$, what allows the cardinality of the set of initial valid triplets $T(G)$, computed by Algorithm~\ref{alg:triplets} (\textit{Triplets(G)}), to be minimized.

%----------------------------------------------------------------------------------------------------------------------------
\begin{center}
 \begin{minipage}{0.7 \textwidth}
  \begin{algorithm2e}[H]
   \DontPrintSemicolon
   \LinesNumbered
   \small

   \BlankLine
   \KwIn{Graph $G$.}
   \KwOut{Set $C$ of all chordless cycles of $G$.}
   \BlankLine

   $G \leftarrow$ DegreeLabeling($G$);\;\nllabel{cc-step1}

   $(T, C) \leftarrow$ Triplets($G$);\;\nllabel{cc-step2}

   %$C \leftarrow \varnothing$.\;\nllabel{cc-step3}
   \BlankLine

   \ForEach{$u \in V$}{\nllabel{cc-step4}
      $blocked(u) \leftarrow 0$.\;\nllabel{cc-step5}
   }

   \BlankLine

   \While{$(T \neq \varnothing)$}{\nllabel{cc-step6}

      $p \leftarrow \langle x, u, y\rangle\in T$;\nllabel{cc-step8}      \tcp*[h]{$p$ is a chordless path.}\;
      $T \leftarrow T - \{p\}$;\;\nllabel{cc-step7}
      \BlankLine
      BlockNeighbors($u$);\;\nllabel{cc-step10}
      %BlockNeighbors($y$);\;\nllabel{cc-step11} %tem que ter para estender a partir da tripla
      
      $C \leftarrow$ CC-Visit($p, C, \ell(u)$, blocked);\;\nllabel{cc-step15}

      \BlankLine
      UnblockNeighbors($u$).\;\nllabel{cc-step17}
   }

   \BlankLine

   \Return $C$.\nllabel{cc-step20}

   \caption{ChordlessCycles($G$) \label{alg:chordless_cycles}}
  \end{algorithm2e}
 \end{minipage}
\end{center}

%----------------------------------------------------------------------------------------------------------------------------
\begin{center}
 \begin{minipage}{0.8 \textwidth}
  \begin{algorithm2e}[H]
%   \DontPrintSemicolon
   \LinesNumbered
   \small

   \BlankLine
   \KwIn{Graph $G$.}
   \KwOut{A labeling of vertices of $G$.}
   \BlankLine

   \ForEach{$v \in V$}{\nllabel{nl-step1}
      	$degree(v) \leftarrow 0$\;\nllabel{nl-step2}
      	$color(v) \leftarrow white$.\nllabel{nl-step3}
	
     	\ForEach{$u \in Adj(v)$}{\nllabel{nl-step4}
		$degree(v) \leftarrow degree(v) + 1$.\nllabel{nl-step5}
	}
   }

    \BlankLine

   \For{$i = 1$ \KwTo $n$}{\nllabel{nl-step6}
        $min\_degree \leftarrow n$.\nllabel{nl-step7}
        
     	\ForEach{$x \in V$}{
		\If{$((color(x) = white)\ \mathbf{and}\ (degree(x) < min\_degree))$}{\nllabel{nl-step8}
			$v \leftarrow x$\;\nllabel{nl-step9}
			$min\_degree \leftarrow degree(x)$.\nllabel{nl-step10}
		}
	}
         \BlankLine
	$\ell(v) \leftarrow i$\;\nllabel{nl-step11}
	$color(v) \leftarrow black$.\nllabel{nl-step12}
         \BlankLine
      	 \ForEach{$u \in Adj(v)$}{\nllabel{nl-step13}
		\If{$color(u) = white$}{\nllabel{nl-step14}
			$degree(u) \leftarrow degree(u) - 1$.\nllabel{nl-step15}
		}
	 }
   }

   \BlankLine

   \Return $\ell$.\nllabel{nl-step16}

   \caption{DegreeLabeling($G$) \label{alg:new_labeling}}
  \end{algorithm2e}
 \end{minipage}
\end{center}

%-----------------------------------------------------------------------------------------------------------------
\begin{center}
 \begin{minipage}{0.8 \textwidth}
  \begin{algorithm2e}[H]
   \DontPrintSemicolon
   \LinesNumbered
   \small
   \caption{Triplets($G$) \label{alg:triplets}}

   \BlankLine
   \KwIn{Undirected simple graph $G$.}
   \KwOut{Set $T(G)$ of initial chordless paths and set $C$ of cycles of length 3.}
   \BlankLine
   $T(G) \leftarrow \varnothing$.\;\nllabel{t-step1}
   $C \leftarrow \varnothing$.\;\nllabel{t-step2}
   \BlankLine

   \ForEach{$u \in V$}{\nllabel{t-step3}
       \BlankLine

       \tcp*[h]{Generate all triplets on form $\langle x,u,y\rangle$}.

       \ForEach{$x, y \in Adj(u)\ \mathbf{such\ that}\ \ell(u) < \ell(x) < \ell(y)$}{\nllabel{t-step4}
           \eIf{$(x, y) \notin E$}{\nllabel{t-step5}
               $T(G) \leftarrow T(G) \cup \{\langle x, u, y\rangle\}$.\nllabel{t-step6}
           }
           {
           	   $C \leftarrow C \cup \{\langle x, u, y\rangle\}$.\nllabel{t-step7}	
           }
       }
   }

   \BlankLine
   \Return $(T(G), C)$.\nllabel{t-step8}
  \end{algorithm2e}
 \end{minipage}
\end{center}

%-----------------------------------------------------------------------------------------------------------------
\begin{center}
 \begin{minipage}{0.75 \textwidth}
  \begin{algorithm2e}[H]
   \DontPrintSemicolon
   \LinesNumbered
   \small
   \caption{CC\_Visit($p, C, key, blocked$) \label{alg:cc_visit}}

   \BlankLine
   \KwIn{Path $p = \langle u_1, u_2, \dots, u_t\rangle$ such that $p$ is a chordless path; set $C$ of chordless cycles; $key = \ell(u_2)$, that is the least value of this chordless path; and global array $blocked$.}
   \KwOut{Set $C$ of chordless cycles.}
   \BlankLine

   BlockNeighbors($u_{t}$).\nllabel{ccv-step1}

   \ForEach{$v \in Adj(u_t)$}{\nllabel{ccv-step2}
      \If{$((\ell(v) > key)\ \mathbf{and}\ (blocked(v) = 1))$}{
	  $p' \leftarrow \langle p, v\rangle$;\;\nllabel{ccv-step2-2}
    %pode fazer a lista de adjacÃªncias de trÃ¡s para frente, para economizar tempo na busca.
	  \eIf{$((v, u_1) \in E)$}{\nllabel{ccv-step3}
	      $C \leftarrow C \cup \{ p'\}$;\;\nllabel{ccv-step4}
	  }
	  {
        $C \leftarrow$ CC-Visit($p', C, key$);\;\nllabel{ccv-step6}
	  }
     }
   }

  \BlankLine
  UnblockNeighbors($u_t$).\;\nllabel{ccv-step7}
  \BlankLine
  \Return $C$.\nllabel{ccv-step8}

  \end{algorithm2e}
 \end{minipage}
\end{center}

%-----------------------------------------------------------------------------------------------------------------
\begin{center}
 \begin{minipage}{0.75 \textwidth}
  \begin{algorithm2e}[H]
   \DontPrintSemicolon
   \LinesNumbered
   \small
   \caption{BlockNeighbors($v, blocked$) \label{alg:block_neighbors}}

   \BlankLine
   \KwIn{A vertex $v \in V$ and a globlal array $blocked$.}
   \KwOut{Blockade of all vertices on neighborhood of $v$.}
   \BlankLine

   \ForEach{$u \in Adj(v)$}{\nllabel{bn-step1}
      $blocked(u) \leftarrow blocked(u) + 1$.\;\nllabel{bn-step2}
   }

  \BlankLine

  \end{algorithm2e}
 \end{minipage}
\end{center}

%-----------------------------------------------------------------------------------------------------------------
\begin{center}
 \begin{minipage}{0.75 \textwidth}
  \begin{algorithm2e}[H]
   \DontPrintSemicolon
   \LinesNumbered
   \small
   \caption{UnblockNeighbors($v, blocked$) \label{alg:unblock_neighbors}}

   \BlankLine
   \KwIn{A vertex $v \in V$ and global array $blocked$.}
   \KwOut{Unblockade of all vertices on neighborhood of $v$.}
   \BlankLine

   \ForEach{$u \in Adj(v)$}{\nllabel{un-step1}
	\If{$(blocked(u) > 0)$}{\nllabel{un-step2}
		$blocked(u) \leftarrow blocked(u) - 1$.\;\nllabel{un-step3}
	}
   }

  \BlankLine

  \end{algorithm2e}
 \end{minipage}
\end{center}
\end{document}